\documentclass[submission,copyright,creativecommons]{eptcs}
\providecommand{\event}{TARK 2021} % Name of the event you are submitting to
\usepackage{breakurl}             % Not needed if you use pdflatex only.
\usepackage{underscore}           % Only needed if you use pdflatex.
\usepackage{mathrsfs}
\usepackage{subfigure}
\usepackage[english]{babel}
\usepackage{wrapfig}
\usepackage{amsmath}
\usepackage{amsthm}
\usepackage{amssymb}

\usepackage{color}

\newtheorem{theorem}{Theorem}

\newtheorem{definition}{Definition}
\newtheorem{lemma}{Lemma}

\newtheorem{prop}{Proposition}

\newcommand{\forset}{\mathsf{F}}

\newcommand{\argb}{\mathsf{B}}

\newcommand{\krel}{\mathcal{R}}

\newcommand{\sdef}{\rightsquigarrow}
\newcommand{\model}{(W,\belset, \sow, \rulset,\nf, ||\cdot||)}

\newcommand{\ground}{GE(AF^{M})}

\newcommand{\black}{\color{black}}

\newcommand{\defeat}{\mathsf{defeat}}
\newcommand{\ra}{\rangle}
\newcommand{\la}{\langle}

\newcommand{\defrule}{\mathsf{DefRule}}

\newcommand{\trule}{\mathsf{TopRule}}
\newcommand{\strict}{\mathsf{strict}}

\newcommand{\rebut}{\mathsf{rebuts}}

\newcommand{\po}{\mathsf{P1}}
\newcommand{\pt}{\mathsf{P2}}

\newcommand{\oms}{\mathsf{At}}

\newcommand{\sto}{\mathsf{\twoheadrightarrow}}
\newcommand{\well}{\mathsf{wellshap}}

\newcommand{\lanab}{\mathcal{L}}
\newcommand{\seq}{\mathsf{SEQ}(\forset)}
\newcommand{\ow}{\mathsf{aware}}
\newcommand{\undercut}{\mathsf{undercuts}}
\newcommand{\accept}{\mathsf{accept}}
\newcommand{\sub}{\mathsf{sub_A}}
\newcommand{\prem}{\mathsf{Prem}}
\newcommand{\conc}{\mathsf{conc}}
\newcommand{\sconc}{\mathsf{Conc}}
\newcommand{\tr}{\mathsf{TopRule}}
\newcommand{\belset}{\mathcal{B}}
\newcommand{\rulset}{\mathcal{D}}
\newcommand{\kmodel}{(W,\krel, \sow, \rulset,\nf, ||\cdot||)}

\newcommand{\rul}{R}
\newcommand{\schrul}{((\varphi_1,...,\varphi_n),\varphi)}

\newcommand{\sow}{\mathcal{O}}
\newcommand{\nf}{\mathfrak{n}}
\newcommand{\argset}{\mathsf{A}}
\newcommand{\kt}{\models_{k}}
\newcommand{\pre}{\mathsf{pre}}

  \newcommand{\welset}{WS}
 \newcommand{\kmclass}{\mathcal{K}}
  \newcommand{\canmod}{S^{c}}
\newcommand{\ang}[1]{\langle #1 \rangle}
\newcommand{\strictalpha}{\ang{\alpha_1,...,\alpha_n \sto \varphi}}
\newcommand{\defalpha}{\ang{\alpha_1,...,\alpha_n \Rightarrow \varphi}}

\makeindex

\usepackage{tikz}
\usetikzlibrary{positioning,arrows,calc,shapes,babel,fit} 
\tikzset{
modal/.style={>=stealth',shorten >=1pt,shorten <=1pt,auto,node distance=1.5cm,
semithick},
world/.style={circle,draw,minimum size=0.5cm,fill=gray!15},
point/.style={circle,draw,inner sep=0.5mm,fill=black},
reflexive above/.style={->,loop,looseness=7,in=120,out=60},
reflexive below/.style={->,loop,looseness=7,in=240,out=300},
reflexive left/.style={->,loop,looseness=7,in=150,out=210},
reflexive right/.style={->,loop,looseness=7,in=30,out=330}
}

\title{An Awareness Epistemic Framework for\\ Belief, Argumentation and Their Dynamics}
\author{Alfredo Burrieza \qquad\qquad Antonio Yuste-Ginel
\institute{Department of Philosophy, \\University of M\'alaga, Spain}
\email{\quad burrieza@uma.es \quad\qquad antonioyusteginel@gmail.com}
}
\def\titlerunning{A Logic for Belief, Argumentation and their Dynamics.}
\def\authorrunning{A. Burrieza \& A. Yuste-Ginel}
\begin{document}
\maketitle

\begin{abstract}
The notion of argumentation and the one of belief stand in a problematic relation to one another. On the one hand, argumentation is crucial for belief formation: as the outcome of a process of arguing, an agent might come to (justifiably) believe that something is the case. On the other hand, beliefs are an input for argument evaluation: arguments with believed premisses are to be considered as strictly stronger by the agent to arguments whose premisses are not believed. An awareness epistemic logic that captures qualified versions of both principles was recently proposed in the literature. This paper extends that logic in three different directions. First, we try to improve its conceptual grounds, by depicting its philosophical foundations, critically discussing some of its design choices and exploring further possibilities. Second, we provide a (heretofore missing) completeness theorem for the basic fragment of the logic. Third, we study, using techniques from dynamic epistemic logic, how different forms of information change can be captured in the framework.
\end{abstract}
\sloppy

\section{Introduction}
Belief and argumentation are two central dimensions of humans' cognitive architecture. They have received attention from antiquity to nowadays, and from a broad range of disciplines. It is then unsurprising that formal researchers have undertaken the task of modelling both phenomena. Regarding beliefs, there is an important amount of options for capturing some of its formal aspects \cite{sep-formal-belief}. These models usually capture \emph{what} kind of things are believed (typically, propositions or sentences); \emph{who} believes them (intelligent agents); and, only sometimes, \emph{how} strong or safe these beliefs are (for instance, in probabilistic models of belief or in plausibility structures \cite{baltag2008}). However, most of them fail to capture \emph{why} agents do believe certain things. This lack motivates the recent trial within the epistemic logic community of capturing the missing justification component. This enterprise has been approached from a variety of methods: justification logic \cite{artemov2005,artemov2012,artemov2018awareness,artemov-fitting-sep}, evidence logics based on neighbourhood semantics \cite{van2012evidence,van2014evidence}, and its further topological development  \cite{baltag2016justified}, amongst others. Yet another natural candidate to model justification consists in using conceptual and technical tools coming from argumentation theory (as done, e.g. in \cite{grossi2014,shi2017argument,liandxu2020,comma}). \par
As to argumentation theory, it is a well-established, interdisciplinary field of research  \cite{hbargthe}. Since the last few decades, formal argumentation has gained more and more attention within the field of artificial intelligence, and its general advantages have been highlighted several times \cite{bench2007argumentation,addedvalue}. Within formal approaches to argumentation, it is frequent to distinguish between \emph{abstract approaches} (those that consider arguments as primitive, atomic entities) and \emph{structured approaches} (those that explicitly account for the structure of arguments). For expository purposes, we just mention the popular Dung's approach to abstract argumentation \cite{dung1995acceptability}, based on so-called \emph{abstract argumentation frameworks}, and the ASPIC family of formalisms for structured argumentation, e.g., ASPIC$^{+}$ \cite{prakken2013,aspichbchapter}, that will be the main argumentative resources used in this paper.
\par 
Recently, some works have taken the first steps to explore and exploit the relations among the two different traditions (epistemic logic and formal argumentation). These can be divided in two groups. On the one hand, there are works using epistemic logic tools to reason about argumentation frameworks \cite{schwarzentruber2012building,sakama2020epistemic,proietti2021}. On the other hand, there are works using argumentation tools to provide an (argumentatively inspired) notion of justified belief (the already mentioned \cite{grossi2014,shi2017argument,liandxu2020,comma}). The current paper is inserted in the latter group, and it follows the ideas of \cite{comma} that, contrarily to \cite{grossi2014,shi2017argument,liandxu2020}, and according to more standard ideas in structured argumentation, decides to model arguments as \emph{syntactic entities}.
\par
We start by pointing out that the informal relation between argumentation and belief is itself problematic. Arguably, there is a tension between two intuitive principles governing belief formation and argument evaluation. These principles are:
\begin{description}
\item[$\po$\label{po}] \emph{Beliefs are an input for argument evaluation}, meaning that arguments with believed premisses are better to those with contingent or even rejected premisses.\footnote{We use the term \emph{contingent} in its doxastic sense, that is, a sentence is said to be contingent iff it is neither believed nor believed to be false.}
\item[$\pt$\label{pt}] \emph{Argumentation is an input for belief formation}, meaning that rational agents should believe sentences that are ground in good arguments.
\end{description}

The mentioned tension arises when one tries to embrace both principles without any restriction, leading to an infinite regress. A very similar problem can be found in the root of a long-standing debate about the structure of epistemic justification within contemporary epistemology. \emph{Foundationalist} solutions to such a tension, to which we adhere here, consists in distinguishing between \emph{basic (non-inferred) beliefs} and \emph{non-basic (inferred) beliefs}, where the latter inherit the justification from the former \cite{sep-justep-foundational}. This implies accepting qualified version of both principles, but giving some sort of priority to \nameref{po} over \nameref{pt}. Curiously enough, an analogous distinction can be found as one of the basis of the recent argumentative theory of reason advocated by Mercier and Sperber \cite{sperber2011}. In this context, basic beliefs are called \emph{intuitive beliefs} while inferred beliefs are called \emph{reflective beliefs} (see \cite{sperber1997intuitive}  for a detailed exposition of the distinction). \par
In the rest of this paper, we follow up the work made in \cite{comma}, by extending it in three different directions. First, and after recalling the logic introduced there, whose language allows talking about basic beliefs and structured arguments, we provide its sound and complete axiomatisation (Section \ref{seclog}). We then explain how to use this logic for reasoning about \emph{explicit basic beliefs} and \emph{argument-based belief}, discussing some of the design choices, as well as depicting some alternatives (Section \ref{sec:bels}). Finally, we extend the basic fragment of the logic so as to capture different kinds of informational dynamics, illustrating their effects on both types of beliefs (Section \ref{sec:dyn}).
\section{An awareness logic for belief and argumentation}\label{seclog}
Let us start by recalling the logic introduced in \cite{comma}. We follow the traditional order for presentation: syntax, semantics, and proof theory. We assume a countable set of \emph{propositional letters} $\oms$ as fixed from now on. The \textbf{language} $\lanab$ is defined as the the pair $(\forset,\argset)$ of \emph{formulas} and \emph{arguments} which are respectively generated by the following grammars: 
\begin{center}
 $\varphi::= p \mid \lnot \varphi \mid (\varphi \land \varphi) \mid \square \varphi \mid \ow (\alpha) \mid \conc(\alpha)=\varphi \mid$\\[2mm] $  \mid \strict (\alpha) \mid \undercut(\alpha,\alpha) \mid \well( \alpha) \qquad  p \in \oms, \alpha \in \argset\text{.}$\\[3mm]
 $\alpha::= \langle \varphi \rangle \mid \langle \alpha_1,...,\alpha_n \sto \varphi \rangle \mid \langle \alpha_1,...,\alpha_n \Rightarrow \varphi \rangle  \qquad \varphi \in \forset,  n\geq 1\text{.}$
\end{center}

\par 
The rest of Boolean operators ($\to, \lor, \leftrightarrow$) and constants ($\top,\perp$), as well as the dual of $\square$ (noted $\lozenge$), are defined as usual. Arguments of $\lanab$ have the following informal readings. $\ang{\varphi}$ is an \emph{atomic argument}. Note that this kind of arguments are rather strange in real-life examples, since they have one sole premise and conclusion, and there is not a proper inference step. Mathematically, they can be understood as a one-line proof from $\varphi$ to $\varphi$. As for $\strictalpha$ (resp. $\defalpha$), it represents an argument claiming that $\varphi$ follows deductively (resp. defeasibly) from the conclusions of arguments $\alpha_1,...,\alpha_n$. As an example of a complex argument consider $\ang{\ang{\ang{\mathsf{Bird}},\ang{\mathsf{Bird}\to \mathsf{Wings}} \sto \mathsf{Wings}} \Rightarrow \mathsf{Flies}}$ that informally reads ``This has wings, because it is a bird and all birds have wings. Moreover, since it has wings, it presumably (defeasibly) flies''. \par
Regarding formulas, elements of $\oms$ represent factual, atomic propositions. $\square \varphi$ means that the agent implicitly (ideally) believes that $\varphi$. $\ow(\alpha)$ reads ``the agent is aware of $\alpha$''. $\conc(\alpha)=\varphi$ reads the ``conclusion of $\alpha$ is $\varphi$''. $\strict(\alpha)$ means that $\alpha$ does not contain defeasible inference steps. $\undercut(\alpha,\beta)$ means that $\alpha$ undercuts $\beta$, that is, $\alpha$ attacks some defeasible inference link of $\beta$. Finally, $\well(\alpha)$ means that $\alpha$ has been constructed properly, that is, all its deductive inference steps are valid and all its defeasible inference steps are accepted by the agent.
\par
We use $\seq$ to denote the \emph{set of all finite sequences over $\forset$}. We denote an arbitrary sequence of $n$+1 elements over $\forset$ as $((\varphi_1,...,\varphi_n),\varphi)$. Sequences of formulas are useful to represent inference steps in the meta-language. Although strongly connected from a conceptual point of view, the sequence $((\varphi_1,...,\varphi_n),\varphi)$ is \emph{not} the same object as, for instance, the object language argument $\ang{\ang{\varphi_1},...,\ang{\varphi_n} \Rightarrow \varphi}$. Let $\rul=((\varphi_1,...,\varphi_n),\varphi)\in \seq$ we use $\alpha^{R}$ as a shorthand for $\ang{\ang{\varphi_1},...,\ang{\varphi_n}\Rightarrow \varphi}$. We can see $\alpha^{\rul}$ as the \emph{simplest argument using $\rul$}. As an example, consider the rule $\rul_1= ((\mathsf{Wings}), \mathsf{Flies})$, we have $\alpha^{\rul_1}=\ang{\ang{\mathsf{Wings}}\Rightarrow \mathsf{Flies}}$, but note that the are infinitely many other arguments using $\rul_1$, for instance $\ang{\ang{\ang{\mathsf{Bird}},\ang{\mathsf{Bird}\to \mathsf{Wings}} \sto \mathsf{Wings}} \Rightarrow \mathsf{Flies}}$.
\black \par 
Let us define the following meta-syntactic functions for analysing an \textbf{argument's structure}, taken form ASPIC$^{+}$ \cite{prakken2013}:

$\prem(\alpha)$ returns the \textbf{premisses} of $\alpha$ and it is defined as follows: $\prem(\langle \varphi \rangle):=\{\varphi\}$,  \quad $\prem(\la \alpha_1,...,\alpha_n \hookrightarrow \varphi \ra):= \prem(\alpha_1)\cup ... \cup \prem (\alpha_n)$ where $\hookrightarrow \in \{\sto, \Rightarrow\}$.  
\par
$\sconc(\alpha)$ returns the \textbf{conclusion} of $\alpha$ and it is defined as follows $\sconc(\la \varphi \ra):=\{\varphi\}$ and $\sconc(\la \alpha_1,...,\alpha_n \hookrightarrow \varphi \ra):=\{\varphi\}$ where $\hookrightarrow \in \{\sto, \Rightarrow\}$.
\par
 $\sub(\alpha)$ returns the \textbf{subarguments} of $\alpha$ and it is defined as follows: $\sub(\la \varphi \ra):=\{\la \varphi \ra \}$ and $\sub(\la \alpha_1,...,\alpha_n \hookrightarrow \varphi \ra):=\{\la \alpha_1,...,\alpha_n \hookrightarrow \varphi \ra\}\cup \sub(\alpha_1)\cup...\cup \sub (\alpha_n)$ where $\hookrightarrow \in \{\sto, \Rightarrow\}$. \par 
 $\tr(\alpha)$ returns the \textbf{top rule} of $\alpha$, i.e. the last rule applied in the formation of $\alpha$. It is defined as follows: $\tr(\langle \varphi \rangle)$ is left undefined, $\tr(\la \alpha_1,...,\alpha_n \sto \varphi \ra)=\tr(\la \alpha_1,...,\alpha_n \Rightarrow \varphi \ra):= $ $((\sconc(\alpha_1),...,\sconc(\alpha_n)), \varphi)$. \par 
$\defrule(\alpha)$ returns the set of \textbf{defeasible rules} of $\alpha$ and it is defined as $\defrule(\langle \varphi \ra):=\emptyset$, $\defrule(\la \alpha_1,...,\alpha_n \sto \varphi \ra):=\defrule(\alpha_1)\cup...\cup \defrule(\alpha_n)$ and $\defrule (\la \alpha_1,...,\alpha_n \Rightarrow \varphi \ra):= \{((\sconc(\alpha_1),...,\sconc (\alpha_n)), \varphi)\}\cup \defrule(\alpha_1) \cup... \cup \defrule(\alpha_n)$. 
\par \medskip

 Let us also define \textbf{semantic propositional negations}, for any $\varphi,\psi \in \forset$: $\varphi=\sim \psi$ abbreviates $\well(\ang{\ang{\varphi}\sto \lnot \psi}) \land \well(\ang{\ang{\psi}\sto \lnot \varphi})$. 

\par \smallskip
Let us now move to semantics. A \textbf{model} for $\lanab$ is a tuple $M=\model$ where:
\begin{itemize}
\item $W\neq \emptyset$ is a set of \emph{possible worlds}.
\item $\belset \subseteq W $ and $\belset \neq \emptyset$ is the set of \emph{worlds that are doxastically indistinguishable for the agent}.

\item $\sow \subseteq \argset$ is the set of \emph{available arguments}, also called the \emph{awareness set of the agent}.
\item $\rulset \subseteq \seq$ is a set of \emph{accepted defeasible rules}. Moreover, for every $((\varphi_1,...,\varphi_n),\varphi)\in \rulset$ we require that:
\begin{itemize}

 \item  $\{\varphi_1,...,\varphi_n,\varphi\}\nvdash_0 \perp$ (defeasible rules are consistent), where $\vdash_{0}$ denotes the consequence relation of classical propositional logic, and 
 
 \item $\{\varphi_1,...,\varphi_n\}\nvdash_0 \varphi$ (defeasible rules are not deductively valid).
\end{itemize}

 \item $\nf:  \seq \black \to \oms$ is a (possibly partial) \emph{naming function} for rules, where $\nf(R)$ informally means ``the rule $R$ is applicable''.
 \item $||\cdot||$ is and an \emph{atomic valuation}, i.e. a function $||\cdot||:\oms \to \wp(W)$.

  \end{itemize}

\noindent \textbf{Interpretation.} In a given model $M=\model$, $\sow$ represents the set of arguments that the agent entertains or is aware of. Whenever $\alpha \in \sow$, we assume that (i) the agent can determine her doxastic attitude toward the premisses of $\alpha$ through non-inferential methods (for instance, through observations), and (ii) she knows the structure of $\alpha$ (either because $\alpha$ has been communicated to her, or because she has gone through the cognitive process of building $\alpha$). Besides this, there is not semantic intuition underlying $\sow$, so the agent can be perfectly aware of rather silly arguments, as $\ang{\ang{p}\sto q}$, without accepting them in any sense. Moreover, rules in the set $\rulset$ are interpreted as rules whose inference strength lies in their content, rather than as purely formal schemas (as deductive rules are). As an example, consider the rule ``Peter's bike is on the bike parking area, therefore he should be in his office''. The term \emph{accepted} means that the agent considers them applicable if there are not good reasons against doing so. Note that $\alpha^{\rul} \in \sow$ does not imply $\rul \in \rulset$ (informally corresponding to the intuition that an agent can be aware of a defeasible argument without accepting its rule). There are further restrictions that could be arguably adopted, but that we leave out for the sake of simplicity. For instance, we could require $\sow$ to be closed under subarguments, or that for any accepted defeasible rule, the agent is aware of at least an argument using it.
\par \black
 Let $M=\model$ be a model for $\lanab=(\forset,\argset)$. The \textbf{set of well-shaped arguments} $\welset^{M} \subseteq \argset$ (depending on $\rulset$ in $M$) is  the smallest set  fulfilling the following conditions: 
  \begin{enumerate}
  \item $\langle \varphi \rangle\in \welset^{M}$ for any $ \varphi\in \forset \black$.
 \item  $\la \alpha_1,...,\alpha_n \sto \varphi \ra\in \welset^{M}$ iff both $\alpha_i\in \welset^{M}$ for every $1\leq i \leq n$ and $\{\sconc(\alpha_1),...,\sconc(\alpha_n)\}\vdash_0 \varphi$.
\item $\la \alpha_1,...,\alpha_n \Rightarrow \varphi \ra\in \welset^{M}$ iff  both $\alpha_i\in \welset^{M}$ for every $1\leq i \leq n$ and  $((\sconc(\alpha_1),...,\sconc(\alpha_n)),\varphi) \in \rulset$. 
  \end{enumerate}
  We drop the superscript $M$ from $\welset^{M}$ whenever there is no danger of confusion.

\par \medskip
Let $(M,w)$ be a pointed model for $\lanab$, that is, $M=(W,\belset, \sow, \rulset,\nf, ||\cdot||)$ is a model and $w \in W$. The \textbf{truth} relation, relating pointed models and formulas, is given by:\footnote{Note that we do not need to consider $\undercut$ as a primitive operator, since it could be defined through a (simpler) operator that captures the meaning of $\nf$. We make this choice for the sake of succinctness, as well as for studying the axiomatic behaviour of $\undercut$.} 
\smallskip
\begin{center}

\begin{tabular}{r c l}
$M,w \models \square \varphi$ & iff & for all $w'\in W$: $w' \in \belset$ implies $M,w'\models \varphi$.\\ 
$M,w \models \ow(\alpha)$ & iff & $\alpha \in \sow$. \\
 $M,w\models \conc(\alpha)=\varphi$ & iff & $\sconc(\alpha)=\varphi$. \\
 $M,w\models \strict (\alpha)$ & iff & $\defrule(\alpha)=\emptyset$. \\
 $M,w \models \undercut(\alpha,\beta)$ & iff & $\beta=\ang{\beta_1,...,\beta_n \Rightarrow \psi}$ and $\sconc(\alpha)= \lnot \nf(\trule(\beta))$.\\
 $M,w \models \well(\alpha)$ & iff &  $\alpha\in \welset^{M}$.
\end{tabular}
\end{center}

\smallskip

A formula $\varphi $ is said to be \textbf{valid} (noted $\models \varphi$) iff it is true at all pointed models. We use $||\varphi||_{M}$ to denote the \textbf{truth-set} of $\varphi$, i.e., the set of worlds of $M$ where $\varphi$ is true, and $\mathcal{M}$ to denote the class of all models. \par  
We now present a sound and complete axiomatisation of $\lanab$ w.r.t. $\mathcal{M}$, a topic that was left out in \cite{comma}, and that constitutes one of the main technical contributions of the current paper. Although our models provide a compact representation of the needed components for reasoning about basic and argument-based beliefs in a single-agent context, they are rather non-standard from a technical point of view. Besides the strongly syntactic character of some of their elements, their modal components are not defined as usual, therefore the definition of its canonical model cannot be extrapolated straightforwardly. Nevertheless, we can provide an indirect completeness proof (see Appendix \hyperref[app:compstat]{A1} for details).
\begin{theorem}\label{compstatic} The axiom system $\mathsf{L}^{\mathsf{BA}}$, defined in Table \ref{tab:axioms}, is sound a complete for $ \lanab$ w.r.t. $\mathcal{M}$.
\end{theorem}

%------- TABLA DE AXIOMAS

\black
\begin{table}[h!]
\black
\centering
\sloppy
\begin{tabular}{l}
\hline
 \qquad \textbf{Modal core axioms} \\\hline
(Ax0) \quad All propositional tautologies  \\
(Ax1) \quad $KD45$ axioms for $\square$   \\
\qquad \textbf{Introspection axioms} \\\hline
(Ax2) \quad $ \ow(\alpha) \to \square \ow(\alpha)$ \\
 (Ax3) \quad  $ \lnot \ow(\alpha) \to \square \lnot \ow(\alpha)$ \\
(Ax4) \quad $ \well(\alpha) \to \square \well(\alpha)$ \\
(Ax5) \quad $ \lnot \well(\alpha) \to \square \lnot \well(\alpha)$\\
(Ax6) \quad  $ \undercut(\alpha,\beta)\to \square \undercut(\alpha,\beta)$ \\  
(Ax7) \quad  $ \lnot \undercut(\alpha,\beta)\to \square \lnot \undercut(\alpha,\beta)$ \\
\black
 \qquad \textbf{Axioms for syntactic operators} \\\hline
(Ax8) \quad $ \conc(\alpha)=\varphi$ \quad whenever $\sconc(\alpha)=\varphi$ \\
(Ax9) \quad $ \lnot \conc (\alpha)=\varphi$ \quad whenever $\sconc(\alpha)\neq \varphi$ \\
(Ax10) \quad$ \strict (\alpha)$ \quad whenever $\defrule(\alpha)=\emptyset$ \\
(Ax11) \quad $ \lnot \strict (\alpha)$ \quad whenever $\defrule(\alpha)\neq\emptyset$ \\
\qquad \textbf{Wellshapedness axioms} \\
\hline 
(Ax12) \quad $ \well(\la \varphi \ra)$   \\
(Ax13) \quad$ \well(\la \alpha_1,...,\alpha_n \sto \varphi\ra)\to \bigwedge_{1 \leq i \leq n} \well(\alpha_i)$   \\

(Ax14) \quad $  \bigwedge_{1 \leq i \leq n} \well(\alpha_i) \to \well(\la \alpha_1,...,\alpha_n \sto \varphi\ra)$  \\ \hspace*{\fill} whenever $\{\sconc(\alpha_i)\mid 1 \leq i \leq n\} \vdash_{0} \varphi$ \\
(Ax15) \quad $  \lnot \well(\la \alpha_1,...,\alpha_n \sto \varphi\ra)$ \quad  whenever $\{\sconc(\alpha_i)\mid 1 \leq i \leq n\}\nvdash_{0} \varphi$ \\
%\olive $  \lnot \well(\la \la \varphi_1\ra,...,\la \varphi_n\ra \Rightarrow \varphi\ra)$  whenever $\{\varphi_1,...,\varphi_n,\varphi\}_{0} \perp$ \\

 (Ax16) \quad $ \Big(\bigwedge_{1 \leq i \leq n} \well(\alpha_i) \land \well(\la \la \sconc(\alpha_1)\ra,...,\la \sconc(\alpha_n)\ra \Rightarrow \varphi \ra) \Big)$  \\ \hspace*{\fill} $\leftrightarrow  \well(\la \alpha_1,...,\alpha_n \Rightarrow \varphi \ra) $ \\

%--- AXIOM: consistency of defeasible rules.

(Ax17) \quad  $(\well(\la \alpha_1,...,\alpha_n \Rightarrow \varphi \ra) \to \lnot \well(\la \la \sconc(\alpha_1)\ra ,...,\la \sconc(\alpha_n)\ra, \la \varphi \ra \sto \perp \ra)$\\

(Ax18) \quad $(\well(\la \alpha_1,...,\alpha_n\Rightarrow \varphi \ra) \to \lnot \well(\la \alpha_1,...,\alpha_n   \sto \varphi\ra)$
%------ AXIOM: nf is a function.
\\
\qquad \textbf{Undercut axioms} \\
\hline 
 (Ax19) \quad $ \undercut(\la \lnot p \ra, \alpha^{R})\to \lnot \undercut(\la \lnot q \ra, \alpha^{R})$ \quad whenever $q\neq p$%,  where $\alpha^{R}=\la \la \varphi_1\ra,...,\la \varphi_n\ra \Rightarrow \varphi \ra$ 

\\
(Ax20) \quad $ \lnot \undercut(\alpha, \la \varphi \ra)$ \\
(Ax21) \quad $ \lnot \undercut(\alpha,\la \alpha_1,...,\alpha_n \sto \varphi \ra)$ \\

%%---- AXIOM: Range of nf=negative literals

(Ax22) \quad $ \lnot \undercut(\alpha,\beta)$ \quad whenever $\sconc(\alpha)\neq \lnot p$ for some $p \in \oms$ \\
\black

(Ax23)\quad $ (\undercut(\la \lnot p \ra, \la \la \sconc(\beta_1)\ra ,...,\la \sconc(\beta_n)\ra \Rightarrow \varphi \ra) \land \conc(\alpha)=\lnot p)\to $ \\
\hspace*{\fill} $\undercut(\alpha, \la \beta_1,...,\beta_n \Rightarrow \varphi \ra)$\\
 
(Ax24) \quad $  (\undercut(\alpha, \la \beta_1,...,\beta_n \Rightarrow \varphi \ra) \land \conc(\alpha)=\lnot p)\to $ \\
\hspace*{\fill} $\undercut(\la \lnot p \ra, \la \la \sconc(\beta_1)\ra ,...,\la \sconc(\beta_n)\ra \Rightarrow \varphi \ra)$ \\

\qquad \textbf{Rules} \\ \hline 
(MP) \quad From $ \varphi \to \psi$ and $\varphi$ infer $\psi$  \\
(Nec) \quad From $\varphi$, infer $\square \varphi$ \\
  \\
   %From $\vdash  \bigwedge_{1\leq i \leq n} \well(\beta_i)$, $\vdash \bigwedge_{1\leq i \leq n} \conc(\alpha_i)=\varphi_i $ and $\vdash \bigwedge_{1\leq i \leq n} \varphi_i \to \varphi $ \\ infer $\vdash \well(\la \alpha_1,...,\alpha_n \sto \varphi\ra)$ \\  
\hline

\end{tabular}
\caption{Axiom system.}
\label{tab:axioms}
\end{table}

\section{Basic beliefs and argument-based beliefs}\label{sec:bels}
The logic introduced above can be used to study a rich repertoire of doxastic attitudes. We start by discussing \emph{basic beliefs}, informally representing those that are not grounded on inferential processes. As mentioned, they can also be understood in terms of \emph{intuitive beliefs}, i.e., those that the agent extracts from a sort of data-base, seen by her as completely trustworthy \cite{sperber1997intuitive}. As usual in awareness epistemic logic, we have two versions of this notion. On the one hand, we have the implicit (ideal) version of basic beliefs, modelled through $\square \varphi$, that suffers from the extensively discussed problem of logical omniscience (see e.g. \cite[Chapter 9]{FH2004}). On the other hand, we have its explicit counterpart, for which we have chosen $\square^{e}\varphi:= \square \varphi \land \ow(\ang{\varphi})$. Note that, like in other logics for implicit and explicit belief, it holds that $\models \square^{e}\varphi \to \square \varphi$. Moreover, under the current semantics, $\square^{e} \varphi$ is equivalent to a schema that resembles another usual option for modelling explicit beliefs (e.g. \cite{velazquez2014dynamic}): $\models \square^{e}\varphi \leftrightarrow \square (\varphi \land \ow(\ang{\varphi}))$.
\par 
Besides basic beliefs, we can also capture in $\lanab$ a sort of deductive-explicit belief. Deductive-explicit beliefs are those rooted in a deductive argument s.t. the agent has a basic belief that all its premisses are true. Formally, and following \cite{baltag2012}, we define \textbf{doxastic argument acceptance} as
\begin{center}
$\accept(\alpha):=\bigwedge_{\varphi \in \prem(\alpha)} \square \varphi$,
\end{center}

and \textbf{deductive-explicit belief} as
\begin{center}
$\mathsf{B}^{\mathsf{D}}(\alpha,\varphi):= \accept(\alpha) \land \ow (\alpha) \land \conc(\alpha)=\varphi \land \strict(\alpha)\land \well(\alpha)$.
\end{center} 

\par \smallskip

Note that $\models \mathsf{B}^{\mathsf{D}}(\alpha,\varphi) \to \square \varphi$ and $\models \square^{e}\varphi  \leftrightarrow \mathsf{B}^{\mathsf{D}}(\ang{\varphi},\varphi) $. The first validity shows that deductive-explicit beliefs are a subset of basic-implicit beliefs. The second one shows that basic-explicit beliefs are an extreme case of deductive-explicit beliefs (those that are rooted in the trivial deduction that goes from $\varphi$ to $\varphi$, i.e., in the atomic argument $\ang{\varphi}$). \par 
Up to now, we have not gone far from the kind of attitudes that are usually discussed in the awareness logic literature (e.g. in \cite{FH1987,van2010dynamics,grossi2009twelve,grossi2015syntactic}). We now take a small detour through argumentation theory in order to define argument-based beliefs. Roughly speaking, argument-based beliefs are grounded in arguments that may involve non-deductive steps. They can be understood, at least to some extent, in terms of the \emph{reflective beliefs} of \cite{sperber1997intuitive}. Recall that we are after formalising the principle \nameref{pt} presented in the introduction: the beliefs of a rational agent should be grounded in good arguments. But, what does it means \emph{good} in this context? Following \cite{beirlaen2018argument}, the very notion of argument strength can be analysed in three different layers or dimensions: the \emph{support dimension} (how strong is the reason given by an argument to accept its conclusion), the \emph{dialectic dimension} (how arguments attack and defeat each other), and the \emph{evaluative dimension} (how the former conflicts are to be solved). \par
 Hence, the first step is to set up a notion of argument strength regarding the \emph{support dimension}. Formally, we seek to define a preference relation among the arguments of $\lanab$ that takes into account \nameref{po} (arguments with believed premisses are to be preferred over arguments with premisses that are not believed). In \cite{comma}, we showed how to do this by splitting all arguments in three preference classes that were based on the basic doxastic attitude of the agent toward the premisses of the arguments. Here, we take a much simpler view, for the sake of brevity, and directly exclude arguments whose premisses are not believed. Both options makes \nameref{pt} dependent on \nameref{po}, since in the process of grounding arguments in beliefs and these in turn in new arguments, we arrive at good arguments that are good just because the agent has a basic belief that all their premisses hold. However, inference links must still play a role when determining the relative strength of two arguments, the simplest principle that can be adopted in this regard is captured by the following binary relation among arguments of $\lanab$:
$\alpha \geq \beta:= \strict (\alpha) \lor \lnot \strict(\beta)$. This relation informally corresponds to the idea that, \emph{ceteris paribus}, deductive arguments are to be preferred to non-deductive ones. \par
Regarding the \emph{dialectic dimension} of argument strength, we capture two forms of argumentative defeat, namely, \emph{undercutting} (attacking a defeasible inference step of any subargument) and \emph{successful rebuttal} (attacking the conclusion of a less or equally preferred subargument). 
 Formally,

\begin{itemize}

\item \textbf{Undercutting a subargument}  $\undercut^{\ast}(\alpha,\beta):=\bigvee_{\beta' \in \sub (\beta)}\undercut(\alpha,\beta')$.
\item \textbf{Unrestricted successful rebuttal} \\  $\mathsf{U}\rebut(\alpha,\beta):= \lnot \strict (\beta) \land  \bigvee_{\beta' \in \sub (\beta)} (\conc(\alpha)=\varphi \land \conc(\beta')=\psi \land \varphi=\sim \psi \land \alpha \geq \beta')$.
\item \textbf{Defeat} $\defeat (\alpha,\beta):= \undercut^{\ast}(\alpha,\beta)\lor\mathsf{U}\rebut(\alpha,\beta) $.

\end{itemize}

 As discussed in the formal argumentation literature, there is a more restrictive alternative for the notion of rebuttal, requiring the top rule of the attacked subargument to be defeasible\footnote{See \cite{yu2018structured} for a discussion about the two possible design choices. Note moreover that the other customary type of attack, i.e. \emph{undermining} (attacking a premise), makes sense only when non-believed premisses are taken into account.}
\begin{itemize}
\item \textbf{Restricted successful rebuttal} \\
$\mathsf{R}\rebut(\alpha,\beta):= \lnot \strict (\beta) \land \bigvee_{\ang{\beta_1,...,\beta_n \Rightarrow \varphi} \in \sub (\beta)} (\conc(\alpha)= \psi \land \varphi=\sim \psi)$.
\end{itemize}
\par 
Argumentation frameworks and their semantics \cite{dung1995acceptability} are the most studied tool to capture the \emph{evaluative dimension} of argument strength. We now explain how to incorporate them in the current approach. Let $(M,w)$ be a pointed model for $\lanab=(\forset,\argset)$, we define its \textbf{associated argumentation framework} as $AF^{M}:=(\argset^{M},\sdef)$, where $\argset^{M}:=\{\alpha \in \argset \mid M,w\models \ow(\alpha)\land \well(\alpha) \land \accept (\alpha)\}$ and $\sdef \subseteq \argset^{M}\times \argset^{M}$ is given by $\alpha \sdef \beta$ iff $M,w\models \defeat (\alpha,\beta)$. We stress the fact that in the domain of our frameworks (i.e., in $\argset^{M}$), basic beliefs act as a filter (in the clause $\accept(\alpha)$), instantiating a qualified, unproblematic version of \nameref{po}, namely $\po^{\prime}$: \emph{basic beliefs} are an input for argument evaluation. Given a set of possibly conflicting arguments (an argumentation framework), we need a mechanism for the agent to decide which of the arguments are to be selected (an argumentation semantics). We say that a set $B \subseteq \argset^{M}$ is \textbf{conflict-free} iff there are no $\alpha,\beta \in B$ s.t. $\alpha \sdef \beta$. Moreover, we say that $B\subseteq \argset^{M}$ \textbf{defends} $\alpha \in \argset^{M}$ iff for every $\gamma \in \argset^{M}$, $\gamma \sdef \alpha$ implies that there is $\beta \in B$ s.t. $\beta \sdef \gamma$. We say that $B \subseteq \argset^{M}$ is a \textbf{complete extension} iff it is conflict-free and it contains precisely the elements of $\argset^{M}$ that it defends. We say that $B \subseteq	 \argset^{M}$ is the \textbf{grounded extension} of $AF^{M}=(\argset^{M},\sdef)$ iff it is the smallest (w.r.t. set inclusion) complete extension. We use $\ground$ to denote the grounded extension of $AF^{M}$. As it is well-known, the grounded extension of an argumentation framework always exists and it is moreover unique \cite{dung1995acceptability}. The unfamiliar reader is referred to \cite{baroni2018abstract} for an extensive discussion on argumentation semantics. \par 
Finally, we use the grounded extension to define the \textbf{argument-based beliefs} of the agent. First, let us extend $\lanab=(\forset,\argset)$ to $\lanab^{\mathsf{AB}}=(\forset^{\mathsf{AB}},\argset)$ by adding a new kind of formulas $\mathsf{B}(\alpha,\varphi)$ where $\alpha\in \argset$ and $\varphi \in \forset$. $\mathsf{B}(\alpha,\varphi)$ means that the agent believes that $\varphi$ based on argument $\alpha$. We interpret the new language in the same class of models, by adding the truth clause:
\begin{center}

\begin{tabular}{r c l}

$M,w\models \mathsf{B}(\alpha,\varphi)$ & $\text{iff}$ & $ \alpha \in \ground \quad \text{and} \quad \sconc(\alpha)=\varphi\text{.}$
\end{tabular}
\end{center}
Note that $\models \argb^{\mathsf{D}}(\alpha,\varphi)\to \argb (\alpha,\varphi)$ and $\models \square^{e}\varphi \to \argb (\ang{\varphi},\varphi)$.
 \par\smallskip We close this section by analysing our notion of argument-based belief under the view of \cite{caminada2007evaluation}'s \emph{rationality postulates}. In a nutshell, if no restrictions are imposed, our agent behaves according to a kind of \emph{minimal rationality} (i.e. she does not explicitly believe in inconsistencies). If, however, we add some ideal assumptions, then she satisfies all \cite{caminada2007evaluation}'s postulates.

\begin{prop}\label{prop:rat} Let $(M,w)$ be a pointed model for $\lanab=(\forset,\argset)$, where $M=\model$. Let $AF^{M}$ be its associated argumentation framework, then:
\begin{itemize}
\item $AF^{M}$ satisfies \emph{direct consistency}, that is, there are no $\alpha,\beta \in \argset$ and $\varphi,\psi \in \forset$ s.t. $M,w\models \mathsf{B}(\alpha,\varphi)\land \mathsf{B}(\beta,\psi)\land \varphi=\sim \psi$.

\item If restricted rebuttal is assumed and $\sow=\argset$, then $AF^{M}$ satisfies \emph{direct consistency}; \emph{indirect consistency} (that is, $\sconc(\ground)\nvdash_{0}\perp)$
); \emph{sub-argument closure} (that is, $\alpha \in \ground$ implies $\sub(\alpha)\subseteq \ground$); and \emph{strict closure} (that is, $\sconc(\ground)\vdash_{0}\varphi$ implies $\varphi \in \sconc(\ground)$).

\end{itemize}

 \end{prop}
\begin{proof}[Proof (sketched)] For the first item, we suppose the contrary, that is, that there are arguments $\alpha,\beta \in \ground$ s.t. $\sconc(\alpha)=\varphi$, $\sconc(\beta)=\psi$ and $\varphi$ is propositionally equivalent to the negation of $\psi$. Then, we continue by cases on the shape of $\alpha$ and $\beta$ (each of them can be either an atomic argument, or an argument whose last inference step is deductive (resp. defeasible)). From the nine different cases, three of them are redundant. From the six remaining cases, it is easy to arrive to $\alpha \sdef \beta$ or $\beta \sdef \alpha$ (which contradicts the assumption that they are in the grounded extension, because it is conflict-free). \par 
For the second item, it suffices to show that under both assumptions (adopting the definition of restricted rebuttal and assuming $\sow=\argset$), we are just working with an instance of \emph{well-defined} ASPIC$^{+}$ frameworks (one constructed over a knowledge base where the set of ordinary premisses is empty), which is guaranteed to satisfy all \cite{caminada2007evaluation}'s rationality postulates (see \cite[Section 3.3]{aspichbchapter} for details).\end{proof}
\section{Dynamics of information}\label{sec:dyn}
 The current framework can throw some light on the relations between dynamics of information, argumentation and doxastic attitudes. We can distinguish several kinds of actions, that have different potential effects on basic and argument-based beliefs. The framework naturally allows for the use of tools imported from \emph{dynamic epistemic logic} (DEL) \cite{hans2007}. In particular, we can describe these actions using dynamic modalities, for which complete axiomatisations can be then provided by finding a full list of \emph{reduction axioms} \cite{kooi2007expressivity,hans2007,wang2013axiomatizations}. In order to do so, one first need to show that the rule of replacement of proved equivalents is sound (it preserves validity) in the extended language (see \cite{kooi2007expressivity} for details). Although this is \emph{not} the case in $\lanab$, as it happens with other languages containing awareness operators \cite{FH1987,grossi2009twelve}, we can restrict the domain of application of the rule, and it still does the job for axiomatizing certain dynamic extensions. More precisely, we will work with the rule: 
 \begin{description}
\centering
 \item[(RE)\label{re}] From $\varphi \leftrightarrow \psi$, infer $\delta \leftrightarrow \delta [\varphi/\psi]$,
 \end{description}
  with $\delta[\varphi/\psi]$ the result of replacing one or more non-$\star$ occurrences of $\psi$ in $\delta$ by $\varphi$.\footnote{A non-$\star$ of $\psi$ in $\delta$ is just an occurrence of $\psi$ in $\delta$ where $\psi$ is not inside the scope of $\star \in \{\ow, \conc, \well,\undercut\}$. Note that we assume that $\varphi$ is inside the scope of $\conc$ in the formula $\conc(\alpha)=\varphi$.}
 Semantically, this amounts to showing that each of the actions $\mathsf{act}$ we are about to discuss is well defined, in the sense that whenever we compute $M^{\mathsf{act}}$ (the result of executing action $\mathsf{act}$ in model $M$), we stay in the intended class of models. When this does not happen, as it is the case with many DEL actions (e.g. public announcements \cite{hans2007}), one needs to find a set of \emph{preconditions} for the action. Preconditions works as sufficient conditions for the action to be ``safe'' i.e., to secure that after executing it, we stay in the intended class of models. \par

 Let us start by defining four different actions. Let $\lanab=(\forset,\argset)$ be given, let $M = (W,\belset, \sow, \rulset,\nf, ||\cdot||)$ be an $\lanab$-model, let $\alpha \in \argset$, let $\rul\in \seq$, and let $\varphi \in \forset$. We define:

\begin{itemize}

 \item The act of \textbf{acquiring argument $\alpha$} (resp. \textbf{forgetting argument $\alpha$}) produces the model $  M^{\alpha+!}:=(W,\belset, \sow^{\alpha+!}, \rulset,\nf, ||\cdot||)$, where $\sow^{\alpha+!}:=\sow \cup\{\alpha\}$ (resp. $  M^{\alpha-!}:=(W,\belset, \sow^{\alpha-!}, \rulset,\nf, ||\cdot||)$, where $\sow^{\alpha-!}:=\sow \setminus\{\alpha\}$).
 \item The act of \textbf{accepting the defeasible rule $\rul$} produces the model $M^{\rul +!}:=(W,\belset, \sow, \rulset^{\rul +!},\nf, ||\cdot||)$, where $\rulset^{\rul +!}:=\rulset \cup\{R\}$.
 \item The act of \textbf{publicly announcing} $\varphi$ produces the model $M^{\varphi!}:=(W^{\varphi !},\belset^{\varphi!}, \sow, \rulset,\nf, ||\cdot||^{\varphi!})$, where $W^{\varphi !}:= W \cap||\varphi||_{M}$; $\belset^{\varphi!}:=\belset \cap||\varphi||_{M}$; and $||p||_{M}^{\varphi !}:= ||p||_{M}\cap ||\varphi||_{M}$ for every atom $p$.
 \end{itemize} 
 
 \iffalse
\begin{figure}
\centering
 \includegraphics[scale=1]{beliefs.png}
 \caption{Typology of beliefs with inclusion relation.}
  \end{figure}
  \fi
\paragraph{Interpretation.}  Note that the definition is far from being exhaustive, we analyse them because they are natural adaptations of other actions studied in the literature \cite{grossi2009twelve, hans2007}. The most basic argumentative change we can think of consists in adding an argument into the awareness of the agent. Informally, this can be thought as the result of a communicative event (e.g. an opponent advancing an argument), learning (the agent reading an argument in a book), or as the result of reflection (the own agent constructing an argument). Formally, the action is a direct generalization of the ``consider'' action defined for sentences in \cite{grossi2009twelve,van2010dynamics}. Its straightforward counterpart is the act of forgetting an argument (i.e. dropping it from the awareness of the agent). As for the action $(\cdot)^{\rul+!}$, defeasible rules can also be learnt in different ways. For instance, an agent can learn the rule $((\mathsf{Bird}),\mathsf{Flies})$ because an ornithologist told her, because she observed repeatedly that birds fly, or because she read it in a textbook. Finally, public announcements are probably the most studied action in DEL (see e.g. \cite[Chapter 4]{hans2007}). This kind of announcements are supposed to be truthful and coming from a completely reliable source.\par \smallskip

We now define a \textbf{dynamic language}, in order to talk about the different actions. Let $\lanab=(\forset,\argset)$ be a language, formulas of the extended language $\lanab^{!}=(\forset^{!},\argset)$ are given by:
$$\varphi::=\psi \mid \lnot \varphi \mid (\varphi \land \varphi) \mid [\alpha+!] \varphi \mid [\alpha-!] \varphi\mid [\rul +!]\varphi \mid [\psi!] \varphi \qquad  \psi \in \forset,\alpha\in \argset,\rul \in \seq\text{.}$$

Let $[\mathsf{act}]$ be any of the dynamic modalities we have just defined, we use $\ang{\mathsf{act}}$ as an abbreviation of $\lnot[\mathsf{act}]\lnot$, with $\ang{\mathsf{act}} \varphi$ informally meaning that action $\mathsf{act}$ can be executed and after executing it, $\varphi$ holds. \par 

Note that the class of all models $\mathcal{M}$ is \emph{not} closed under all defined actions. In particular, it is not closed under $(\cdot)^{\rul+!}$ nor under $(\cdot)^{\varphi!}$. For the former, the reason is that only rules that are consistent and non-deductive can be learnt as defeasible (see the definition of model in Section \ref{seclog}). For the latter, only truthful formulas that do not trivialize the beliefs of the agent (in the sense of making $\belset$ empty), can be announced. This inconvenience is solved by fixing \textbf{preconditions} (expressible in $\lanab$) for both actions. Let $\rul=\schrul\in \seq$ and $\varphi \in \forset$, we define:

\par

 \noindent $\pre(\rul):=\lnot \well(\ang{\ang{\varphi_1},...,\ang{\varphi_n},\ang{\varphi}\sto \perp}) \land \lnot \well(\ang{\ang{\varphi_1},...,\ang{\varphi_n} \sto \varphi}) $; and
 \\
  $\pre(\varphi !):=\varphi \land \lozenge \varphi$.
\par
It is almost immediate to check that, for any pointed model $(M,w)$, any $\rul=\schrul  \in \seq$, and any $\varphi \in \forset$ we have that: 
\begin{center}

($\{\varphi_1,...,\varphi_n,\varphi\}\nvdash_{0}\perp$ and $\{\varphi_1,...,\varphi_n\}\nvdash_{0}\varphi$) iff $M,w\models \pre(\rul)$;
and \\
($w \in ||\varphi||_{M}$ and $||\varphi||_{M}\cap \belset \neq \emptyset)$ iff $M,w\models \pre(\varphi !)$. 
\end{center}
 Moreover, note that $M,w\models \pre(\rul)$ iff $M,u\models \pre(\rul)$ for every $u \in W$. Let $(M,w)$ be a pointed model with $M=\model$, we define the truth clause for the new kind of formulas: \smallskip
 \begin{center}

 \begin{tabular}{r c l}
 
$M,w\models [\sigma+!]\varphi$ & iff  & $M^{\sigma+!},w\models \varphi$, \\
$M,w\models [\sigma-!]\varphi$ & iff  & $M^{\sigma-!},w\models \varphi$, \\
$M,w\models [\rul +!]\varphi$ & iff & $M,w\models \pre(\rul)$, implies $M^{\rul+!},w\models \varphi$, \\
$M,w\models [\varphi!] \psi$ &iff & $M,w\models \pre(\varphi!)$ implies $M^{\varphi!},w \models \psi\text{.}$ \\
\end{tabular}
 \end{center}
 \smallskip
 Finally, we establish a completeness result for $\lanab^{!}$ w.r.t. $\mathcal{M}$. Note that in Table \ref{tab:redaxgen}, $\pm$ denotes an arbitrary element of $\{+,-\}$.
  \begin{prop}\label{propcompdyn}The proof system $\mathsf{L}^{!}_{\mathsf{BA}}$ that extends the one of Table \ref{tab:axioms} with all axioms of Table \ref{tab:redaxgen} and it is closed under \nameref{re} is sound and complete for $\lanab^{!}$ w.r.t. $\cal M$. \end{prop}
 \begin{proof}
 Soundness follows from the validity of all axioms and the validity-preserving character of \nameref{re} in the extended language. Completeness follows from the usual reduction argument. In short, note that in the right-hand side of all axioms of Table \ref{tab:redaxgen}, either the dynamic operator disappears or it is applied to a less complex formula than in the left-hand side. In the case of reduction axioms for $[\rul!+] \well(\alpha)$, either there are no dynamic modalities occurring in the right-hand side of the equivalence or they are applied to $\well$-formulas with less complex \emph{arguments} than in the right-hand side. Therefore, we can define a meaning-preserving translation from $\forset^{!}$ to $\forset$ that, together with Theorem \ref{compstatic}, provides the desired result. The validity-preserving character of \nameref{re} in the extended language w.r.t. $\mathcal{M}$ takes care of formulas with nested dynamic modalities. The reader is referred to \cite{kooi2007expressivity} for details.  \par 
 \end{proof}
 We close this section by modelling a toy example, inspired by \cite{toulmin1958}, and illustrating how actions affect argument-based beliefs.
 Suppose that an agent is wondering whether another agent, Harry, is a British subject ($\mathsf{br}$). Suppose that the only basic-explicit belief she holds at the beginning is that \emph{Harry was born in Bermuda} ($\mathsf{be}$). Other pieces of relevant information are: \emph{Harry's parents are aliens} ($\mathsf{a}$), and that \emph{the rule ``If Harry is born in Bermuda, then he is presumably a British subject'' is applicable} ($\mathsf{r1}$). Let $R_1=((\mathsf{be}),\mathsf{br})$. We start with the model $M_0=\model$, where $W=\belset=\{w_0,w_1,w_2,w_3\}$, $\rulset=\emptyset$, $\sow=\{\ang{\mathsf{be}}\}$, $\nf(R_1)=\mathsf{r1}$, $||\mathsf{be}||=W$, $||\mathsf{br}||=||\mathsf{r1}||=\{w_0,w_2\}$, and $||\mathsf{a}||=\{w_0,w_1\}$. It is then easy to check that $M_0,w_0\models \square^{e} \mathsf{be}$. Moreover, we have that $M_0,w_0\models [R1+!][\alpha^{R_1}+!]\mathsf{B}(\alpha^{R_1},\mathsf{br})$. In words, after learning the rule $R_1$ and becoming aware of the simplest argument using it, i.e. $\ang{\ang{\mathsf{be}}\Rightarrow \mathsf{br}}$, the agent has an argument-based belief that Harry is a British subject. If, however the agent learns subsequently from a completely trustworthy source that Harry's parents are alien ($\mathsf{a}$), together with the rule $R_2=((\mathsf{a}),\lnot \mathsf{r1})$, and the argument $\ang{\ang{\mathsf{a}} \Rightarrow \lnot \mathsf{r1}}$, then she revises her argument-based belief about Harry's nationality. In symbols,
$M_0,w_0\models [R1+!][\alpha^{R_1}+!][\mathsf{a}!][R_2+!][\alpha^{R_2}+!]\lnot \mathsf{B}(\alpha^{R_1},\mathsf{br})$.
\begin{table}[h!]
\centering
\begin{tabular}{l | l}
\hline
 $[\alpha\pm!]p \leftrightarrow p$ & $[\varphi!]p \leftrightarrow (\pre(\varphi !)\to p)$ \\
 $[\alpha\pm!]\lnot \varphi \leftrightarrow \lnot [\alpha\pm!] \varphi $ & $[\varphi!]\lnot \psi \leftrightarrow (\pre(\varphi !)\to \lnot [\varphi!] \psi)$ \\
$[\alpha\pm!](\varphi \land \psi) \leftrightarrow ([\alpha\pm!] \varphi \land [\alpha\pm!] \psi)$ &  $[\varphi!](\delta \land \psi) \leftrightarrow ( [\varphi!] \delta \land [\varphi!] \psi)$\\
 $[\alpha\pm!]\square \varphi \leftrightarrow \square[\alpha\pm!]\varphi  $ & $[\varphi!]\square \psi \leftrightarrow (\pre(\varphi !)\to \square[\varphi!] \psi)$ \\
$[\alpha\pm!]\ow(\beta) \leftrightarrow\ow(\beta)$ for $\alpha\neq \beta$ & $[\varphi!] \ow(\beta)\leftrightarrow (\pre(\varphi !)\to \ow(\beta))$ \\
$[\alpha +!]\ow(\alpha) \leftrightarrow\top$ & \\
$[\alpha -!]\ow(\alpha) \leftrightarrow\perp$ & \\
 $[\alpha\pm!]\conc(\beta) =\varphi \leftrightarrow\conc(\beta) =\varphi $ &  $[\varphi!] \conc(\beta)=\psi \leftrightarrow (\pre(\varphi !)\to \conc(\beta)=\psi)$ \\
$[\alpha\pm!]\strict(\beta) \leftrightarrow\strict(\beta) $ & $[\varphi!]\strict(\beta)\leftrightarrow (\pre(\varphi !)\to \strict(\beta))$ \\
 $[\alpha\pm!]\undercut(\beta, \gamma)\leftrightarrow\undercut(\beta, \gamma)$ &  $[\varphi!]\undercut(\beta,\gamma)\leftrightarrow (\pre(\varphi !)\to \undercut(\beta,\gamma))$\\
 $[\alpha\pm!]\well(\beta)  \leftrightarrow\well(\beta)$ & $[\varphi!]\well(\beta)\leftrightarrow (\pre(\varphi !)\to \well(\beta))$ \\ 
\\
\hline 

\end{tabular}

%TABULAR FOR REDAXRUL
\begin{tabular}{l l}

  $[\rul+!]p \leftrightarrow (\pre(\rul)\to p)$ & $[\rul+!] \ow(\alpha)\leftrightarrow (\pre(\rul)\to \ow(\alpha))$ \\ 
  $[\rul+!]\lnot \varphi \leftrightarrow (\pre(\rul)\to \lnot [\rul+!] \varphi)$ & $[\rul+!] \conc(\alpha)=\varphi \leftrightarrow (\pre(\rul)\to \conc(\alpha)=\varphi)$ \\
 $[\rul+!](\varphi \land \psi) \leftrightarrow ( [\rul+!] \varphi\land [\rul+!] \psi)$ & $[\rul+!]\strict(\alpha)\leftrightarrow (\pre(\rul)\to \strict(\alpha))$  \\
  $[\rul+!]\square \varphi \leftrightarrow \square[\rul+!] \varphi$ & $[\rul+!]\undercut(\alpha,\beta)\leftrightarrow (\pre(\rul)\to \undercut(\alpha,\beta))$ \\

& \\

 \multicolumn{2}{l}{$[\rul+!]\well(\ang{\varphi})\leftrightarrow \top$} \\
  \multicolumn{2}{l}{$[\rul +!]\well(\ang{\ang{\varphi_1},...,\ang{\varphi_n}\sto \varphi})\leftrightarrow \big(\pre(\rul)\to \well(\ang{\ang{\varphi_1},...,\ang{\varphi_n}\sto \varphi})\big)$} \\
  \multicolumn{2}{l}{$[\rul +!]\well(\alpha^{\rul})\leftrightarrow \top$} \\
  \multicolumn{2}{l}{$[\rul +!]\well(\alpha^{\rul'})\leftrightarrow \big(\pre(\rul)\to\well(\alpha^{\rul'})\big)$ \quad whenever $\rul\neq \rul'$} \\
 \multicolumn{2}{l}{$[\rul+!] \well(\ang{\alpha_1,...,\alpha_n \hookrightarrow \varphi}) \leftrightarrow $} \\
\multicolumn{2}{c}{ $\leftrightarrow \Big( \pre(\rul) \to \big( \bigwedge_{1\leq i \leq n} [\rul +!]\well(\alpha_i) \land [\rul+!]\well(\ang{\ang{\sconc (\alpha_1)},...,\ang{\sconc (\alpha_n)}\hookrightarrow \varphi})\big)\Big)$}\vspace{0.4mm}\\

\hline
\end{tabular}
\caption{Reduction axioms for $\lanab^{!}$.}
\label{tab:redaxgen}
\end{table}

%\begin{ex}

%\end{ex}

\section{Concluding remarks}

\paragraph{Closely related work.} From all the works we have commented throughout the paper, it seems that \cite{grossi2009twelve,grossi2015syntactic} and \cite{shi2017argument} are the closest one to our approach. Regarding \cite{grossi2009twelve,grossi2015syntactic}, we have somehow generalize their \emph{awareness of rules} to our awareness of arguments (abstracting away from other forms of awareness treated there). As for \cite{shi2017argument}, their choice of modelling arguments semantically (as opens of a topology), permits a transparent axiomatisation of their notion of argument-based beliefs, which is easily guaranteed to be consistent (two of the weaknesses of our approach). On the other hand, we naturally treat arguments as first-class citizens in our language, and the argument-based beliefs of our agent escape from every form of logical omniscience (while the beliefs of \cite{shi2017argument}'s agent are still closed under equivalent formulas). \par 
\paragraph{Future work.} There are natural open paths for future work. An urgent task in the development of the logical aspects of the framework consists in axiomatizing (if possible) the argument-based belief operator $\argb(\cdot,\cdot)$. Moreover, the modal semantic apparatus of our models could be extended to plausibility structures \cite{baltag2008}, so as to model fine-grained preference between arguments, based on the agent's basic epistemic attitudes toward the premisses of the involved arguments (e.g. known premisses are to be preferred to strongly believed premisses, and the latter, in turn, are to be preferred to merely believed premisses). Finally, a multi-agent extension of the current framework could be used to model argument exchange in different kinds of scenarios (e.g. deliberation, persuasion dialogues or inquiry).

\black
%\nocite{*}
\bibliographystyle{eptcs}
\bibliography{tesis}

\section*{Appendix (Proof sketch of Theorem \ref{compstatic})}\label{sec:app}
\addcontentsline{toc}{section}{\nameref{sec:app}}

The outline of the proof is as follows: we first define a new class of (non-standard) models for our language (which are Kripke models where the syntactic components 
--awareness, accepted rules and names of rule-- are maintained throughout the accessibility relation). We then show two things: (i) we can go from pointed Kripke models to its generated submodels without loosing $\lanab$-information (just as in the general modal case) and; (ii) we can transform systematically Kripke generated submodels into our models (again, without loosing $\lanab$-information). Finally, we prove completeness w.r.t. the class of non-standard models and apply (i) and (ii) to obtain the desired result. Let us unfold some of the details. \par
 First of all, we define a \textbf{Kripke model for $\lanab=(\forset,\argset)$} as a tuple $S=(W,\krel, \sow, \rulset,\nf, ||\cdot||)$ where: $W\neq \emptyset$ is a set of \emph{possible worlds}; $\krel \subseteq W \times W$ is a serial, transitive and euclidean relation; $\sow:W\to \wp(\argset)$ is a function assigning an awareness set $\sow(w)$ to each world $w$; $\rulset: W\to \wp(\seq)$ (with $n\in \mathbb{N}$) is a function assigning a set of \emph{accepted defeasible rules} $\rulset(w)$ to each world $w$; $\nf: (W\times \seq) \to \oms$ is a (possibly partial) \emph{naming function} for defeasible rules, where $\nf(w,R)$ informally means ``the defeasible rule $R$ is applicable at $w$''; and $||\cdot||:\oms \to \wp(W)$ is a valuation function.  Moreover, we assume that for every $w,w'\in W$, $w\krel w'$ implies $\sow(w)=\sow(w')$, $\rulset(w)=\rulset(w')$, and $\nf(w,R)=\nf(w',R)$. We also assume that if $((\varphi_1,...,\varphi_n),\varphi) \in \rulset(w)$, then $\{\varphi_1,...,\varphi_n,\varphi\}\nvdash_0 \perp$ and $\{\varphi_1,...,\varphi_n\}\nvdash_0 \varphi$.

\par

 Note that now $\welset$ sets depend on both the model and the world at which we are looking (since $\rulset$ may vary from one world to another). Consequently, we use $\welset^{S}(w)$ to denote the set of well-shaped arguments at $(S,w)$.\par  Truth w.r.t. pointed Kripke models is denoted by $\kt$ and defined as follows (the missing clauses are as expected): 

 \begin{center}
 
\begin{tabular}{r c l}

$S,w\kt \square \varphi$ & iff & $w\krel v$ implies $S,v\kt \varphi$ \\
$S,w\kt \ow(\alpha)$ & iff & $\alpha\in \sow(w)$ \\
 $S,w\kt \well(\alpha)$ & iff &$\alpha \in \welset^{S}(w)$ \\
  
  $S,w\kt \undercut(\alpha,\beta)$ & iff & $\beta=\ang{\beta_1,...,\beta_n\Rightarrow \varphi}$ and $\sconc(\alpha)=\lnot \nf (w,\trule(\beta))$.
  
 \end{tabular} 
 \end{center} \par
We say that a Kripke model $S=\kmodel$ is \textbf{uniform} iff for every $w,w'\in W$ it holds that: (i) $\sow(w)=\sow(w')$; (ii) $\rulset(w)=\rulset(w')$; and (iii) $\nf(w,R)=\nf(w',R)$ for every $R\in \rulset$. $\kmclass$ denotes the class of all pointed Kripke models, and $\kmclass^{u}$ denotes the class of all uniform pointed Kripke models. We abuse of notation and use $\mathcal{M}$ to denote the class of all \emph{pointed} models (the standard ones that we defined in Section \ref{seclog}). 
\paragraph{Transformation lemmas.} Now, we need a couple of lemmas. The first one says that we can go from pointed Kripke models to Kripke uniform pointed models without loosing $\lanab$-information, by taking generated submodels. We use $S^{w}$ to denote the submodel of $S$ generated by $w$ (see \cite[Chapter 2]{blackburn2002}).

\begin{lemma}\label{lem:trans1} Let $(S,w)\in \kmclass$. We have that: \begin{itemize}
\item[i)] $(S^{w},w) \in \kmclass^{u}$,

i.e. each pointed-generated submodel of a Kripke model is a uniform Kripke model.
\item[ii)] For every $\varphi \in \forset$, $(S,w)\kt\varphi \quad \text{iff} \quad (S^{w},w)\kt\varphi$, i.e. truth is preserved under generated submodels. \black
\end{itemize} \end{lemma}

Item \emph{i)} follows easily from the definition of generated submodel and uniform Kripke model. Item \emph{ii)} can be proved by induction on $\varphi$.\par 

The second lemma says that we can go from Kripke uniform models to our models (the standard ones, defined in Section \ref{seclog}) without loosing $\lanab$-information.

\begin{lemma}\label{lem:trans2} For every uniform pointed Kripke model $(S,v) \in \kmclass^{u}$, there is a pointed model $(M,w)\in \mathcal{M}$ s.t. for every $\varphi \in \forset$: $$ S,v\kt \varphi \quad \text{iff}\quad M,w\models \varphi\text{.}$$\end{lemma}

Let us define the function $\tau$ for each uniform Kripke model as follows $\tau(S,w)=(\tau(S),\tau(w))$ where $\tau(w)=w$ and $\tau(S)=(\tau(W),\tau(\krel),\tau(\sow),\tau(\rulset),\tau(\nf),\tau(||\cdot||))$ s.t.: 
\begin{center}

\begin{tabular}{c}

$\tau(W):=\{w\}\cup \krel[w]$, \\
$\tau(\krel):=\krel[w]$, \\
$\tau(\sow):=\sow(w)$, \\
$\tau(\rulset):=\rulset(w)$,\\
$\tau(\nf):=\{(R,p) \in \seq \times \oms \mid \nf(w,R)=p\}$, \\
$\tau(||p||):= ||p||\cap \tau(W) $ for every $p\in \oms\text{.}$

\end{tabular}
\end{center} \par \smallskip

Now, it is easy to check that, for every $(S,w)\in \kmclass^{u}$: $\tau((S,w))\in \mathcal{M}$, that is $\tau:\kmclass^{u}\to \mathcal{M}$. Once this is done, we can show that, for every $\varphi \in \forset$, it holds that: 
\begin{center}
$S,w\kt \varphi \quad \text{iff} \quad \tau(S,w)\models \varphi\text{.}$
\end{center}
The proof of the last assertion is by induction on $\varphi$ where the step for $\varphi=\well(\alpha)$ is another inductive argument (on the construction on $\alpha$).
\paragraph{Completeness w.r.t. Kripke models.} We can now define \textbf{the canonical Kripke} model for $\lanab$ as: \begin{center}
$\canmod=(W^{c},\krel^{c}, \sow^{c}, \rulset^{c},\nf^{c}, ||\cdot||^{c})\text{,}$
\end{center}  where the definition of  $W^{c}$, $\krel^{c}$ y $||\cdot||^{c}$ is as usual in modal logic \cite{blackburn2002}; while the definition of the rest of the elements mimics the one of awareness operators \cite{FH1987}:\begin{center}
$\sow^{c}(\Gamma):=\{\alpha \in \argset \mid \ow(\alpha)\in \Gamma\}\text{,}$

$\rulset^{c}(\Gamma):=\{\rul \in \seq \mid  \well(\alpha^{\rul}) \in \Gamma\}\text{,}$

$((\Gamma,\rul) ,p)\in \nf^{c} \quad \text{iff}\quad \undercut(\la \lnot p \ra, \alpha^{\rul})\in \Gamma\text{.} $
\end{center}
 %We use $\welset^{c}(\Gamma)$ to abbreviate $\welset^{\canmod}(\Gamma)$.\black \\
Now, we need to prove:

\begin{lemma}[Canonicity]\label{canlem} $\canmod$ is a Kripke model for $\lanab$. \end{lemma}
For showing that $\canmod$ satisfies all conditions, we reason using maximally-consistent set properties and our axiom system. As illustrations: semantic restrictions on the accessibility relations follows from (Ax1) (see e.g. \cite{FH2004} or \cite{blackburn2002}), while (Ax19) permits showing that $\nf^{c}$ is a function.

\begin{lemma}[Truth]\label{truthlem} For every $\varphi \in \forset$: $\varphi \in \Gamma$ iff $\canmod,\Gamma\kt \varphi$.  \end{lemma}

The proof proceeds by induction on $\varphi$. The Boolean and modal cases are standard \cite{blackburn2002}. The cases for operators $\ow$, $\conc$ and $\strict$ are straightforward (they actually do not make use of the induction hypothesis, due to their syntactic character). The cases for $\varphi=\undercut(\alpha,\beta)$ and $\varphi=\well(\alpha)$ are slightly more compromised. For the latter, another inductive argument on the structure of $\alpha$ is required.  \par

\paragraph{Completeness w.r.t. standard models.} Finally, completeness w.r.t. standard models can be proved as follows. Suppose $\Gamma \nvdash \varphi$, then $\Gamma \cup \{\lnot \varphi\}$ is consistent. By Lindenbaum, we have that there is a $\Gamma^{+}\in W^{c}$ s.t $\Gamma \cup \{\lnot \varphi\}\subseteq \Gamma^{+}$. By the Truth Lemma we have that $\canmod,\Gamma^{+}\kt \Gamma \cup \{\lnot \varphi\}$. By item \emph{ii)} of Lemma \ref{lem:trans1} we have that $S^{c\Gamma^{+}},\Gamma^{+}\kt \Gamma \cup \{\lnot \varphi\}$ and by item \emph{i)} we have that $S^{c\Gamma^{+}},\Gamma^{+}$ is a pointed uniform Kripke model. Then by Lemma \ref{lem:trans2} we know that $\tau(S^{c\Gamma^{+}},\Gamma^{+})\models \Gamma \cup \{\lnot \varphi\}$ which implies by definition of semantic logical consequence that $\Gamma \nvDash \varphi$.
\black

\end{document}